\def\submission{0}
\ifnum\submission=0
    \documentclass{article}
    \usepackage{authblk}
    \usepackage[a4paper,top=3cm,bottom=2.5cm,left=3cm,right=3cm,marginparwidth=1.75cm]{geometry}
\else
    \documentclass{llncs}
\fi

\title{Adaptive versus Static Multi-oracle Algorithms, and Quantum Security of a Split-key PRF\thanks{\textcopyright IACR 2022. Up to a different formatting, this article
		is the final version submitted by the authors to the IACR and to
		Springer-Verlag in Sep 2022. The published version is available from the
		proceedings of  TCC - 2022.}}

\ifnum\submission=0
    \author[1]{Jelle Don}
    \author[1,2]{Serge Fehr}
    \author[1]{Yu-Hsuan Huang}
    \affil[1]{\small Centrum Wiskunde \& Informatica, The Netherlands}
    \affil[2]{\small Leiden University, The Netherlands}
    \affil[{ }]{\small \tt \{jelle.don, serge.fehr, yhh\}@cwi.nl}
    \date{}

    \usepackage[usenames,dvipsnames]{xcolor}
    \usepackage{breakurl}
    \usepackage{amsmath,amssymb,amsthm}
    \newtheorem{theorem}{Theorem}
    \newtheorem{lemma}{Lemma}
    \newtheorem{corollary}{Corollary}

\else
    \author{
        Jelle Don\inst{1}\and
        Serge Fehr\inst{1,2}\and
        Yu-Hsuan Huang\inst{1}
    }
    \institute{
        Centrum Wiskunde \& Informatica (CWI), Amsterdam, The Netherlands\and
        Mathematical Institute, Leiden University, Leiden, The Netherlands\\
        \email{\{jelle.don, serge.fehr, yhh\}@cwi.nl}
    }

    \usepackage{amsmath,amssymb,amsthm}
\fi

\usepackage{algcompatible}
\usepackage{hyperref}

\usepackage{tikz}

\usepackage[normalem]{ulem}

\newcommand{\Kcal}{{\cal K}}
\newcommand{\Xcal}{{\cal X}}
\newcommand{\Ycal}{{\cal Y}}
\newcommand{\Wcal}{{\cal W}}
\newcommand{\Ecal}{{\cal E}}

\newcommand{\Ocal}{{\cal O}}

\newcommand{\tr}{{\sf tr}}

\newcommand{\Acal}{{\cal A}}
\newcommand{\Bcal}{{\cal B}}
\newcommand{\Hcal}{{\cal H}}

\DeclareMathAlphabet\mathbfcal{OMS}{cmsy}{b}{n}

\ifnum\submission=0
\theoremstyle{remark}
\newtheorem{remark}{Remark}
\fi

\begin{document}
    \maketitle
    \begin{abstract}
        In the first part of the paper, we show a generic compiler that transforms any oracle algorithm that can query multiple oracles {\em adaptively}, i.e., can decide on {\em which} oracle to query at what point dependent on previous oracle responses, into a {\em static} algorithm that fixes these choices at the beginning of the execution. Compared to naive ways of achieving this, our compiler controls the blow-up in query complexity for each oracle {\em individually}, and causes a very mild blow-up only. 
        
        In the second part of the paper, we use our compiler to show the security of the very efficient hash-based {\em split-key PRF} proposed by Giacon, Heuer and Poettering (PKC~2018), in the {\em quantum} random-oracle model. 
        Using a split-key PRF as the key-derivation function gives rise to a secure KEM combiner. Thus, our result shows that the hash-based construction of Giacon {\em et al.}\ can be safely used in the context of quantum attacks, for instance to combine a well-established but only classically-secure KEM with a candidate KEM that is believed to be quantum-secure. 
        
        Our security proof for the split-key PRF crucially relies on our adaptive-to-static compiler, but we expect our compiler to be useful beyond this particular application. Indeed, we discuss a couple of other, known results from the literature that would have profitted from our compiler, in that these works had to go though serious complications in order to deal with adaptivity. 
    \end{abstract}

\section{Introduction}
    
    This paper offers two main contributions. In a first part, we show a generic reduction from adaptive to static multi-oracle algorithms, with a mild increase of the query complexity {\em for each oracle individually}, and in the second part, exploiting the reduction from the first part, we prove quantum security of the hash--based split-key pseudorandom function (skPRF) proposed in~\cite{giacon2018kem}. We now discuss these two contributions in more detail. 
    
\paragraph{\bf Adaptive versus Static Multi-Oracle Algorithms.}
    
In certain cryptographic security games, the attacker $\Acal$ is an {\em oracle} algorithm that is given query access to {\em multiple} oracles. This is in particular the case when considering the design of a cryptographic scheme in an idealized setting. Consider for instance the security definitions of public-key encryption and signature schemes in the (quantum) random-oracle model, where the attacker is given oracle access to both: the random-oracle and to a decryption/signing oracle. 
    
By default, such an attacker $\Acal$ can then choose {\em adaptively}, i.e., depending on answers to previous queries, at what point to query {\em which oracle}.
This is in contrast to a {\em static} $\Acal$ that has a predefined order of when it queries which oracle.%
 \footnote{In either case, we allow $\Acal$ to decide adaptively {\em what input} to query, when having decided (adaptively or statically) on which oracle to query. }
In certain cases, proving security for a static attacker is easier than proving security for a full fledged adaptive attacker, or taking care of adaptivity (naively) results in an unnecessary blow-up in the error term (see later).



In this light, it seems to be desirable to have a generic compiler that transforms any adaptive attacker $\Acal$ into a static attacker $\bar \Acal$ that is equally successful in the attack. And there is actually a simple, naive solution for that. Indeed, let $\Acal$ be an arbitrary oracle algorithm that makes adaptive queries to $n$ oracles ${\Ocal}_1,\ldots,{\Ocal}_n$, and consider the static oracle algorithm $\bar \Acal$ defined as follows: $\bar \Acal$ simply runs $\Acal$, and at every point in time when $\Acal$ makes a query to one of ${\Ocal}_1,\ldots,{\Ocal}_n$ (but due to the adaptivity it will only become clear at the time of the query {\em which} ${\Ocal}_i$ is to be queried then), the algorithm $\bar \Acal$ makes $n$ queries, one to every ${\Ocal}_i$, and it relays $\Acal$'s query to the right oracle, while making dummy queries to the other oracles. 

At first glance, this simple solution is not too bad. It certainly transforms any adaptive $\Acal$ into a static $\bar \Acal$ that will be equally successful, and the blow-up in the total query complexity is a factor $n$ only, which is mild given that the typical case is $n = 2$. However, it turns out that in many situations, considering the blow-up in the total query complexity is not good enough. 

For example, consider again the case of an attacker against a public-key encryption scheme in the random-oracle model. In this example, it is typically assumed that $\Acal$ may make many more queries to the random-oracle than to the decryption oracle, i.e., $q_H \gg q_D$. But then, applying the above simple compiler, $\bar \Acal$ makes the same number of queries to the random-oracle and to the decryption oracle; namely $\bar q_H = \bar q_D = q_H + q_D$. 
Furthermore, the actual figure of merit, namely the advantage of an attacker $\bar \Acal$, is typically not (bounded by) a function of the total query complexity, but a function of the two respective query complexities $q_H$ and $q_D$ {\em individually}. For example, if one can show that the advantage of any {\em static} attacker $\bar\Acal$ with respective query complexities $\bar q_H$ and $\bar q_D$ is bounded by, say, $\bar q_H \bar q_D^2$, then the above compiler gives a bound on the advantage of any {\em adaptive} attacker $\Acal$ with respective query complexities $q_H$ and $q_D$ of $q_H^3 + 2 q_H^2 q_D + q_H q_D^2$. If $q_H \gg q_D$ then this is significantly worse than $\approx q_H q_D^2$, which one might hope for given the bound for static $\bar\Acal$. 

Our first result is a compiler that transforms any {\em adaptive} oracle algorithm $\Acal$ that makes at most $q_i$ queries to oracle ${\Ocal}_i$ for $i=1,\ldots,n$ into a {\em static} oracle algorithm $\bar \Acal$ that makes at most $\bar q_i = n q_i$ queries to oracle ${\cal O}_i$ for $i=1,\ldots,n$. Thus, rather than controling the blow-up in the total number of queries, we can control the blow-up in the number of queries for each oracle {\em individually}, yet still with the same factor $n$. Our result applies for {\em any} vector ${\bf q} = (q_1,\ldots,q_n) \in \mathbb{N}^n$ and contains no hidden constants. Our compiler naturally depends on $\bf q$ (or, alternatively, needs $\bf q$ as input) but otherwise only requires straight-line black-box access to $\Acal$, and it preserves efficiency: the run time of $\bar \Acal$ is polynomial in $Q = q_1+\cdots+q_n$, plus the time needed to run $\Acal$. Furthermore, the compiler is applicable to any classical or quantum oracle algorithm $\Acal$, where in the latter case the queries to the oracles ${\cal O}_1,\ldots,{\cal O}_n$ may be classical or quantum as well; however, the {\em choice} of the oracle for each query is assumed to be classical (so that individual query complexities are well defined).

In the above made-up example of a public-key encryption scheme with advantage bounded by $\bar q_H \bar q_D^2$ for any static $\bar\Acal$ with respective query complexities $\bar q_H$ and $\bar q_D$, we now get the bound $8 q_H q_D^2$ for any adaptive $\Acal$ with respective query complexities $q_H$ and $q_D$. 

We show the usefulness of our adaptive-to-static compiler by discussing two example results from the literature. 
One is the security proof by Alkim {\em et al.}~\cite{RevqTesla} of the qTESLA signature scheme~\cite{qTesla} in the quantum random-oracle model; the other is the recent work by Alagic, Bai, Katz and Majenz \cite{ABKM21} on the quantum security of the famous Even-Mansour cipher. 
In both these works, the adaptivity of the attacker was a serious obstacle and caused a significant overhead and additional complications in the proof. 
With our results, these complications could have been avoided without sacrificing much in the security loss (as would be the case with using a naive compiler). 
We also exploit our adaptive-to-static compiler in our second main contribution, discussed below.

Interestingly, all three example applications are in the realm of quantum security (of a classical scheme). 
This seems to suggest that the kind of adaptivity we consider here is not so much of a hurdle in the case of classical queries.
Indeed, in that case, a typical argument works by inspecting the entire query transcript and identifying an event with the property that conditioned on this event, whatever needs to be shown holds {\em with certainty}, and then it remains to show that this event is very likely to occur. 
In the case of quantum queries, this kind of reasoning does not apply since one cannot ``inspect'' the query transcript anymore; 
instead, one then typically resorts to some sort of hybrid argument where queries are replaced one-by-one, and then adaptivity of the queries may\,---\,and sometimes does, as we discuss\,---\,form a serious obstacle.

\paragraph{\bf Quantum-security of a Split-key PRF.} 

In the upcoming transition to post-quantum secure cryptographic standards, {\em combiners} may play an important role. A combiner can be used compile several crypographic schemes into a new, ``combined'' scheme, which offers the same (or a similar) functionality, and so that the new scheme is secure as long as {\em at least one} of the original schemes is secure. 
For example, combining a well-established but quantum-insecure scheme with a believed-to-be quantum-secure (but less well studied) scheme then offers the best of both worlds: it offers security against quantum attacks, should there really be a quantum computer in the future, but it also offers some protection in case the latter scheme turns out to be insecure (or less secure than expected) even against classical attacks.  In other words, using a combiner in this context ensures that we are not making things less secure by trying to aim for quantum security. 

In~\cite{giacon2018kem}, Giacon, Heuer and Poettering showed that any {\em split-key PRF} (skPRF) gives rise to a secure KEM combiner. In more detail, they show that if a skPRF is used in the (rather) obvious way as a key-derivation function in a KEM combiner, then the resulting combined KEM is IND-CCA secure if at least one of the component KEMs is IND-CCA secure.  
They also suggest a few candidates for skPRFs. The most efficient of the proposed constructions is a hash-based skPRF, which is proven secure in~\cite{giacon2018kem} in the random-oracle model. 
However, in the context of a quantum attack, which is in particular relevant in the above example application of a combiner, it is crucial to prove security in the {\em quantum} random-oracle model~\cite{BonehEtal2011}. 
Here, we close this gap by proving security of the hash-based skPRF construction proposed by Giacon {\em et al.}\ in the quantum random-oracle model. 

Our security proof crucially exploits our adaptive-to-static compiler to reduce a general, adaptive attacker/distinguisher to a static one. 
Namely, in spirit, our security proof is a typical hybrid proof, where we replace, one by one, the queries to the (sk)PRF by queries to a truly random function; however, the crux is that for each hybrid, corresponding to a particular function query that is to be replaced, the closeness of the current to the previous hybrid depends on the number of hash queries {\em between the current and the previously replaced function query}. In case of an adaptive $\Acal$, {\em each} such ``window'' of hash queries between two function queries could be as large as the total number of hash queries in the worst case, giving rise to a huge multiplicative blow-up when using this naive bound. Instead, for a static $\Acal$, each such window is bounded by a fixed number, with the sum of these numbers being the total number of hash queries. 

By means of our compiler, we can turn the possibly adaptive $\Acal$ into a static one (almost) for free, and this way avoid an unnecessary blow-up, respectively bypass additional complications that arise by trying to avoid this blow-up by other means.

\section{Preliminaries}
    
    We consider oracle algorithms $\Acal^{\Ocal_1,\dots,\Ocal_n}$ that make queries to (possibly unspecified) oracles $\Ocal_1,\dots,\Ocal_n$, see~Fig.~\ref{fig:algs} (left). Sometimes, and in particular when the oracles are not specified, we just write $\Acal$ and leave it implicit that $\Acal$ makes oracle calls. We allow $\Acal$ to be classical or quantum, and in the latter case we may also allow the queries (to some of the oracles) to be quantum; however, the choice of {\em which} oracle is queried is always classical. For the purpose of our work, we may assume $\Acal$ to have no input; any potential input could be hardwired into $\Acal$. For a vector ${\bf q} = (q_1,\ldots,q_n) \in \mathbb{N}^n$, we say that $\Acal$ is a {\em $\bf q$-query} oracle algorithm if it makes at most $q_i$ queries to the oracle $\Ocal_i$. 

In general, such an oracle algorithm $\Acal$ may decide {\em adaptively} which oracle to query at what step, dependent on previous oracle responses. In contrast to this, a {\em static} oracle algorithm has an arbitrary but pre-defined order in querying the oracles. 

Our goal will be to transform any {\em adaptive} oracle algorithm $\Acal$ into a {\em static} oracle algorithm $\bar \Acal$ that is functionally equivalent, while keeping the blow-up in query complexity for each individual oracle, i.e., the blow-up for each individual $q_i$, small. By {\em functionally equivalent} (for certain oracle instantiations) we mean the respective executions of $\Acal^{O_1,\dots,O_n}$ and $\bar \Acal^{O_1,\dots,O_n}$ give rise to the same output distribution {\em for all} (the considered) instantiations $O_1,\dots,O_n$ of the oracles $\Ocal_1,\dots,\Ocal_n$. In case of {\em quantum} oracle algorithms, we require the output state to be the same. 

For this purpose, we declare that an {\em interactive oracle algorithm} $\Bcal$ is an interactive algorithm with two distinct interaction interfaces, one for the interaction with $\Acal$ (we call this the {\em simulation interface}), and one for the oracle queries (we call this the {\em oracle interface}), see Fig.~\ref{fig:algs} (middle). For any oracle algorithm $\Acal$, we then denote by $\Bcal[\Acal]$ the oracle algorithm that is obtained by composing $\Acal$ and $\Bcal$ in the obvious way. In other words, $\Bcal[\Acal]$ runs $\Acal$ and answers all of $\Acal$'s oracle queries using its simulation interface; furthermore, $\Bcal[\Acal]$ outputs whatever $\Acal$ outputs at the end of this run of~$\Acal$, see Fig.~\ref{fig:algs} (right).%
\footnote{Note, we silently assume consistency between $\Acal$ and $\Bcal$, i.e.\ $\Acal$ should send a message when $\Bcal$ expects one and the format of these messages should match the format of the messages that $\Bcal$ expects (and vice versa), so that the above composition makes sense. Should $\Bcal$ encounter some inconsistency, it will abort. }

    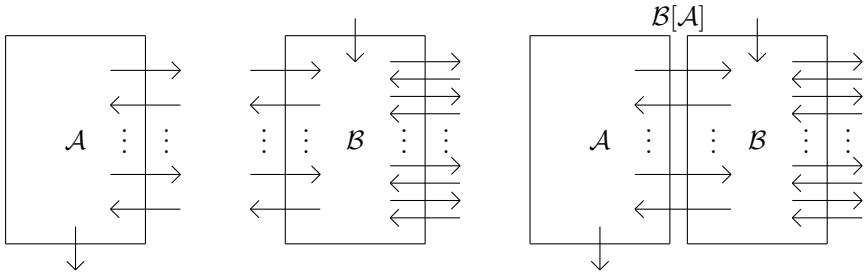
\begin{figure}\quad
        \begin{tikzpicture}[scale=0.46]
        
            \node [] (0) at (-10, 6) {};
            \node [] (1) at (-6, 6) {};
            \node [] (2) at (-10, 0) {};
            \node [] (3) at (-6, 0) {};
            \node [] (4) at (-6, 6) {};
            \node [] (15) at (-8, 3) {};
            \node [] (16) at (-8, 3) {$\Acal$};
            \node [] (24) at (-7, 5) {};
            \node [] (25) at (-5, 5) {};
            \node [] (26) at (-7, 4) {};
            \node [] (27) at (-5, 4) {};
            \node [] (28) at (-7, 2) {};
            \node [] (29) at (-5, 2) {};
            \node [] (30) at (-7, 1) {};
            \node [] (31) at (-5, 1) {};
            \node [] (32) at (-6.6, 3.2) {$\vdots$};
            \node [] (33) at (-5.4, 3.2) {$\vdots$};
            \node [] (34) at (-5.25, 5.25) {};
            \node [] (35) at (-5.25, 4.75) {};
            \node [] (36) at (-6.75, 4.25) {};
            \node [] (37) at (-6.75, 3.75) {};
            \node [] (38) at (-5.25, 2.25) {};
            \node [] (39) at (-5.25, 1.75) {};
            \node [] (40) at (-6.75, 1.25) {};
            \node [] (41) at (-7, 1) {};
            \node [] (42) at (-6.75, 0.75) {};
            \node [] (43) at (-2, 6) {};
            \node [] (44) at (2, 6) {};
            \node [] (45) at (-2, 0) {};
            \node [] (46) at (2, 0) {};
            \node [] (47) at (2, 6) {};
            \node [] (48) at (0, 3) {};
            \node [] (49) at (0, 3) {$\Bcal$};
            \node [] (69) at (-3, 5) {};
            \node [] (70) at (-1, 5) {};
            \node [] (71) at (-3, 4) {};
            \node [] (72) at (-1, 4) {};
            \node [] (73) at (-3, 2) {};
            \node [] (74) at (-1, 2) {};
            \node [] (75) at (-3, 1) {};
            \node [] (76) at (-1, 1) {};
            \node [] (77) at (-2.6, 3.2) {$\vdots$};
            \node [] (78) at (-1.4, 3.2) {$\vdots$};
            \node [] (79) at (-1.25, 5.25) {};
            \node [] (80) at (-1.25, 4.75) {};
            \node [] (81) at (-2.75, 4.25) {};
            \node [] (82) at (-2.75, 3.75) {};
            \node [] (83) at (-1.25, 2.25) {};
            \node [] (84) at (-1.25, 1.75) {};
            \node [] (85) at (-2.75, 1.25) {};
            \node [] (86) at (-3, 1) {};
            \node [] (87) at (-2.75, 0.75) {};
            \node [] (88) at (5, 6) {};
            \node [] (89) at (9, 6) {};
            \node [] (90) at (5, 0) {};
            \node [] (91) at (9, 0) {};
            \node [] (92) at (9, 6) {};
            \node [] (93) at (7, 3) {};
            \node [] (94) at (7, 3) {$\Acal$};
            \node [] (95) at (8, 5) {};
            \node [] (96) at (10.75, 5) {};
            \node [] (97) at (8, 4) {};
            \node [] (98) at (10.75, 4) {};
            \node [] (99) at (8, 2) {};
            \node [] (100) at (10.75, 2) {};
            \node [] (101) at (8, 1) {};
            \node [] (102) at (10.75, 1) {};
            \node [] (103) at (8.4, 3.2) {$\vdots$};
            \node [] (104) at (10.25, 3.2) {$\vdots$};
            \node [] (105) at (10.5, 5.25) {};
            \node [] (106) at (10.5, 4.75) {};
            \node [] (107) at (8.25, 4.25) {};
            \node [] (108) at (8.25, 3.75) {};
            \node [] (109) at (10.5, 2.25) {};
            \node [] (110) at (10.5, 1.75) {};
            \node [] (111) at (8.25, 1.25) {};
            \node [] (112) at (8, 1) {};
            \node [] (113) at (8.25, 0.75) {};
            \node [] (114) at (9.5, 6) {};
            \node [] (115) at (13.5, 6) {};
            \node [] (116) at (9.5, 0) {};
            \node [] (117) at (13.5, 0) {};
            \node [] (118) at (13.5, 6) {};
            \node [] (119) at (11.5, 3) {};
            \node [] (120) at (11.5, 3) {$\Bcal$};
            \node [] (121) at (12.5, 5.25) {};
            \node [] (122) at (14.5, 5.25) {};
            \node [] (123) at (12.5, 4.75) {};
            \node [] (124) at (14.5, 4.75) {};
            \node [] (129) at (12.9, 3.2) {$\vdots$};
            \node [] (130) at (14.1, 3.2) {$\vdots$};
            \node [] (131) at (14.25, 5.5) {};
            \node [] (132) at (14.25, 5) {};
            \node [] (133) at (12.75, 5) {};
            \node [] (134) at (12.75, 4.5) {};
            \node [] (140) at (9.25, 6.5) {$\Bcal[\Acal]$};
            \node [] (141) at (9.25, 6.5) {};
            \node [] (142) at (12.5, 4.25) {};
            \node [] (143) at (14.5, 4.25) {};
            \node [] (144) at (12.5, 3.75) {};
            \node [] (145) at (14.5, 3.75) {};
            \node [] (146) at (14.25, 4.5) {};
            \node [] (147) at (14.25, 4) {};
            \node [] (148) at (12.75, 4) {};
            \node [] (149) at (12.75, 3.5) {};
            \node [] (150) at (12.5, 2.25) {};
            \node [] (151) at (14.5, 2.25) {};
            \node [] (152) at (12.5, 1.75) {};
            \node [] (153) at (14.5, 1.75) {};
            \node [] (154) at (14.25, 2) {};
            \node [] (155) at (12.75, 2) {};
            \node [] (156) at (12.75, 1.5) {};
            \node [] (157) at (12.5, 1.25) {};
            \node [] (158) at (14.5, 1.25) {};
            \node [] (159) at (12.5, 0.75) {};
            \node [] (160) at (14.5, 0.75) {};
            \node [] (161) at (14.25, 1.5) {};
            \node [] (162) at (14.25, 1) {};
            \node [] (163) at (12.75, 1) {};
            \node [] (164) at (12.75, 0.5) {};
            \node [] (165) at (14.25, 2.5) {};
            \node [] (166) at (1, 5.25) {};
            \node [] (167) at (3, 5.25) {};
            \node [] (168) at (1, 4.75) {};
            \node [] (169) at (3, 4.75) {};
            \node [] (170) at (1.4, 3.2) {$\vdots$};
            \node [] (171) at (2.6, 3.2) {$\vdots$};
            \node [] (172) at (2.75, 5.5) {};
            \node [] (173) at (2.75, 5) {};
            \node [] (174) at (1.25, 5) {};
            \node [] (175) at (1.25, 4.5) {};
            \node [] (176) at (1, 4.25) {};
            \node [] (177) at (3, 4.25) {};
            \node [] (178) at (1, 3.75) {};
            \node [] (179) at (3, 3.75) {};
            \node [] (180) at (2.75, 4.5) {};
            \node [] (181) at (2.75, 4) {};
            \node [] (182) at (1.25, 4) {};
            \node [] (183) at (1.25, 3.5) {};
            \node [] (184) at (1, 2.25) {};
            \node [] (185) at (3, 2.25) {};
            \node [] (186) at (1, 1.75) {};
            \node [] (187) at (3, 1.75) {};
            \node [] (188) at (2.75, 2) {};
            \node [] (189) at (1.25, 2) {};
            \node [] (190) at (1.25, 1.5) {};
            \node [] (191) at (1, 1.25) {};
            \node [] (192) at (3, 1.25) {};
            \node [] (193) at (1, 0.75) {};
            \node [] (194) at (3, 0.75) {};
            \node [] (195) at (2.75, 1.5) {};
            \node [] (196) at (2.75, 1) {};
            \node [] (197) at (1.25, 1) {};
            \node [] (198) at (1.25, 0.5) {};
            \node [] (199) at (2.75, 2.5) {};
            \node [] (200) at (-8, 0.5) {};
            \node [] (201) at (-8, -0.75) {};
            \node [] (202) at (-8.25, -0.5) {};
            \node [] (203) at (-7.75, -0.5) {};
            \node [] (204) at (0, 6.5) {};
            \node [] (205) at (0, 5.25) {};
            \node [] (206) at (-0.25, 5.5) {};
            \node [] (207) at (0.25, 5.5) {};
            \node [] (212) at (7, 0.5) {};
            \node [] (213) at (7, -0.75) {};
            \node [] (214) at (6.75, -0.5) {};
            \node [] (215) at (7.25, -0.5) {};
            \node [] (217) at (11.5, 6.5) {};
            \node [] (218) at (11.5, 5.25) {};
            \node [] (219) at (11.25, 5.5) {};
            \node [] (220) at (11.75, 5.5) {};
        
            \draw (0.center) to (4.center);
            \draw (4.center) to (3.center);
            \draw (3.center) to (2.center);
            \draw (2.center) to (0.center);
            \draw (24.center) to (25.center);
            \draw (34.center) to (25.center);
            \draw (35.center) to (25.center);
            \draw (26.center) to (27.center);
            \draw (26.center) to (36.center);
            \draw (26.center) to (37.center);
            \draw (28.center) to (29.center);
            \draw (38.center) to (29.center);
            \draw (39.center) to (29.center);
            \draw (41.center) to (31.center);
            \draw (41.center) to (40.center);
            \draw (41.center) to (42.center);
            \draw (43.center) to (47.center);
            \draw (47.center) to (46.center);
            \draw (46.center) to (45.center);
            \draw (45.center) to (43.center);
            \draw (69.center) to (70.center);
            \draw (79.center) to (70.center);
            \draw (80.center) to (70.center);
            \draw (71.center) to (72.center);
            \draw (71.center) to (81.center);
            \draw (71.center) to (82.center);
            \draw (73.center) to (74.center);
            \draw (83.center) to (74.center);
            \draw (84.center) to (74.center);
            \draw (86.center) to (76.center);
            \draw (86.center) to (85.center);
            \draw (86.center) to (87.center);
            \draw (88.center) to (92.center);
            \draw (92.center) to (91.center);
            \draw (91.center) to (90.center);
            \draw (90.center) to (88.center);
            \draw (95.center) to (96.center);
            \draw (105.center) to (96.center);
            \draw (106.center) to (96.center);
            \draw (97.center) to (98.center);
            \draw (97.center) to (107.center);
            \draw (97.center) to (108.center);
            \draw (99.center) to (100.center);
            \draw (109.center) to (100.center);
            \draw (110.center) to (100.center);
            \draw (112.center) to (102.center);
            \draw (112.center) to (111.center);
            \draw (112.center) to (113.center);
            \draw (114.center) to (118.center);
            \draw (118.center) to (117.center);
            \draw (117.center) to (116.center);
            \draw (116.center) to (114.center);
            \draw (121.center) to (122.center);
            \draw (131.center) to (122.center);
            \draw (132.center) to (122.center);
            \draw (123.center) to (124.center);
            \draw (123.center) to (133.center);
            \draw (123.center) to (134.center);
            \draw (142.center) to (143.center);
            \draw (146.center) to (143.center);
            \draw (147.center) to (143.center);
            \draw (144.center) to (145.center);
            \draw (144.center) to (148.center);
            \draw (144.center) to (149.center);
            \draw (150.center) to (151.center);
            \draw (154.center) to (151.center);
            \draw (152.center) to (153.center);
            \draw (152.center) to (155.center);
            \draw (152.center) to (156.center);
            \draw (157.center) to (158.center);
            \draw (161.center) to (158.center);
            \draw (162.center) to (158.center);
            \draw (159.center) to (160.center);
            \draw (159.center) to (163.center);
            \draw (159.center) to (164.center);
            \draw (165.center) to (151.center);
            \draw (166.center) to (167.center);
            \draw (172.center) to (167.center);
            \draw (173.center) to (167.center);
            \draw (168.center) to (169.center);
            \draw (168.center) to (174.center);
            \draw (168.center) to (175.center);
            \draw (176.center) to (177.center);
            \draw (180.center) to (177.center);
            \draw (181.center) to (177.center);
            \draw (178.center) to (179.center);
            \draw (178.center) to (182.center);
            \draw (178.center) to (183.center);
            \draw (184.center) to (185.center);
            \draw (188.center) to (185.center);
            \draw (186.center) to (187.center);
            \draw (186.center) to (189.center);
            \draw (186.center) to (190.center);
            \draw (191.center) to (192.center);
            \draw (195.center) to (192.center);
            \draw (196.center) to (192.center);
            \draw (193.center) to (194.center);
            \draw (193.center) to (197.center);
            \draw (193.center) to (198.center);
            \draw (199.center) to (185.center);
            \draw (200.center) to (201.center);
            \draw (202.center) to (201.center);
            \draw (201.center) to (203.center);
            \draw (204.center) to (205.center);
            \draw (206.center) to (205.center);
            \draw (205.center) to (207.center);
            \draw (212.center) to (213.center);
            \draw (214.center) to (213.center);
            \draw (213.center) to (215.center);
            \draw (217.center) to (218.center);
            \draw (219.center) to (218.center);
            \draw (218.center) to (220.center);
        
        \end{tikzpicture}
    \caption{An oracle algorithm $\Acal$ (left), an interactive oracle algorithm $\Bcal$ (middle), and the oracle algorithm $\Bcal[\Acal]$ obtained by composing $\Acal$ and $\Bcal$ (right).  }\label{fig:algs}
\end{figure}

In contrast to $\Acal$ (where, for our purpose, any input could be hardwired), we explicitly allow an interactive oracle algorithm $\Bcal$ to obtain an input. Indeed, our transformation, which turns any adaptive oracle algorithm $\Acal$ into a static oracle algorithm $\bar \Acal$, needs to ``know'' $\bf q$, i.e., the number of queries $\Acal$ makes to the different oracles. Thus, this will be provided in the form of an input to $\Bcal$; for reasons to be clear, it be provided in unary, i.e., as $1^{\bf q} := (1^{q_1},\ldots,1^{q_n})$. 

We stress that we do not put any computational restriction on the oracle algorithms $\Acal$ (beyond bounding the queries to the individual oracles); however, we do want our transformation to preserve efficiency. Therefore, we say that an interactive oracle algorithm $\Bcal$ is {\em polynomial-time} if the number of local computation steps it performs is bounded to be polynomial in its input size, and where we declare that copying an {\em incoming} message on the simulation interface to an {\em outgoing} message on the oracle interface, and vice versa, is unit cost (irrespectively of the size of the message).  By providing $\bf q$ in unary, we thus ensure that $\Bcal$ is polynomial-time in $q_1 + \cdots + q_n$.

    \section{A General Adaptive-to-static Reduction for Multi-oracle Algorithms}

    \subsection{Our Result}
    
    Let $n \in \mathbb{N}$ be an arbitrary positive integer. We present here a generic adaptiv-to-static compiler $\cal B$ that, on input a vector ${\bf q} \in \mathbb{N}^n$, turns any {\em adaptive} $\bf q$-query oracle algorithm $\Acal^{\Ocal_1,\ldots,\Ocal_n}$ into a {\em static} $n \bf q$-query algorithm.

    \begin{theorem}\label{thm:reduction}
        There exists a polynomial-time interactive oracle algorithm $\Bcal$, such that for any ${\bf q}\in\mathbb{N}^{n}$ and any {\em adaptive} ${\bf q}$-query oracle algorithm $\Acal^{\Ocal_1,\dots,\Ocal_n}$, the oracle algorithm $\Bcal[\Acal](1^{\bf q})$ is a {\em static} $n{\bf q}$-query oracle algorithm that is functionally equivalent to $\Acal$ for all {\em stateless} instantiations of the oracles $\Ocal_1,\dots,\Ocal_n$. 
    \end{theorem}
    
    \begin{remark}
     As phrased, Theorem~\ref{thm:reduction} applies to oracle algorithms $\Acal$ that have no input. This is merely for simplicity. In case of an oracle algorithm $\Acal$ that takes an input, we can simply apply the statement to the algorithm $\Acal(x)$ that has the input $x$ hardwired, and so argue that Theorem~\ref{thm:reduction} also applies in that case. 
    \end{remark}
    
    \begin{remark}
     $\Bcal[\Acal]$ is guaranteed to behave the same way as $\Acal$ for {\em stateless} (instantiations of the) oracles only. This is become most of the queries that $\Bcal[\Acal]$ makes are actually dummy queries (i.e., queries on a default input and with the response ignored), which have no effect in case of stateless oracles, but may mess up things in case of stateful oracles. Theorem~\ref{thm:reduction} extends to arbitrary stateful oracles if we allow $\Bcal[\Acal]$ to {\em skip} queries instead of making dummy queries (but the skipped queries would still count towards the query complexity).  
    \end{remark}

Given the vector ${\bf q} = (q_1,\dots,q_n) \in \mathbb{N}^{n}$, the core of the problem is to find a {\em fixed} sequence of $\Ocal_i$'s in which each individual $\Ocal_i$ occurs at most $n q_i$ times, and so that {\em every} sequence of $\Ocal_i$'s that contains each individual $\Ocal_i$ at most $q_i$ times can be embedded into the former. We consider and solve this abstract problem in the following section, and then we wrap up the proof of Theorem~\ref{thm:reduction} in Section~\ref{sec:Proof}.

    \subsection{The Technical Core}

    \def\TS#1#2{{\cal P}_{<\infty}(#1 \!\times\! #2)}

    Let  $\Sigma$ be an non-empty finite set of cardinality $n$. We refer to $\Sigma$ as the {\em alphabet}. As is common, $\Sigma^*$ denotes the set of finite strings over the alphabet $\Sigma$. In other words, the elements of $\Sigma^*$ are the strings/sequences $s = (s_1,\ldots,s_\ell) \in \Sigma^\ell$ with arbitrary $\ell \in \mathbb{N}$ (including $\ell = 0$). 
    
    Following standard terminology, for $s  = (s_1,\ldots,s_\ell)$ and $s' = (s'_1,\ldots,s'_m)$ in $\Sigma^*$, the {\em concatenation} of $s$ and $s'$ is the string $s\|s' = (s_1,\ldots,s_\ell,s'_1,\ldots,s'_m)$, and $s'$ is a {\em subsequence} of $s$, denoted $s' \sqsubseteq s$ if there exist integers $1 \leq j_1 < \ldots < j_m \leq \ell$ with $(s_{j_1},\ldots,s_{j_m}) = (s'_1,\ldots,s'_m)$. Such an integer sequence $(j_1,\ldots,j_m)$ is then called an {\em embedding} of $s'$ into $s$.%
    \footnote{We use {\em string} and {\em sequence} interchangeably; however, following standard terminology, there is a difference between a {\em substring} and {\em subsequence}: namely, a substring is a subsequence that admits an embedding with $j_{i+1} = j_i + 1$.     } 
    
    Finally, for a function $q:\Sigma\to\mathbb{N}, \sigma \mapsto q_\sigma$, we say that $s  = (s_1,\ldots,s_\ell) \in \Sigma^*$ has {\em characteristic} (at most) $q$ if $\#\{i \,|\, s_i=\sigma\} = q_\sigma$ ($\leq q_\sigma$) for any $\sigma \in \Sigma$.

    \begin{lemma}[Embedding Lemma]\label{lem:string}
        Let $\Sigma$ be an alphabet of size $n$, and let $q:\Sigma\to\mathbb{N}$, $\sigma \mapsto q_\sigma$. Then, there exists a string $s \in \Sigma^*$ with characteristic $n \cdot q: \sigma \mapsto n \cdot q_\sigma$ such that any string $s'\in\Sigma^*$ with characteristic at most $q$ is a subsequence of~$s$, i.e., $s'\sqsubseteq s$. 
    \end{lemma}

    The idea of the construction of the sequence $s$ is quite simple: First, we evenly distribute $n \cdot q_\sigma$ copies of $\sigma$ within the interval $(0,n]$ by ``attaching'' one copy of $\sigma$ to every point in $(0,n]$ that is an integer multiple of $1/q_\sigma$ (see Fig.~\ref{fig:WalkTheLine}). Note that it may happen that different symbols are ``attached'' to the same point. Then, we walk along the interval from $0$ and $n$ and, one by one, collect the symbols we encounter in order to build up $s'$ from left to right; in case we encounter a point with multiple symbols ``attached'' to it, we collect them in an arbitrary order. 

\begin{figure}
\begin{center}
    \begin{tikzpicture}[scale=0.7]
        
        \node [] (0) at (0, 0) {};
        \node [] (1) at (4, 0) {};
        \node [] (2) at (8, 0) {};
        \node [] (3) at (12, 0) {};
        \node [] (4) at (6, 0) {};
        \node [] (5) at (0, 0.25) {};
        \node [] (6) at (0, -0.25) {};
        \node [] (7) at (4, 0.25) {};
        \node [] (8) at (4, -0.25) {};
        \node [] (9) at (6, 0.25) {};
        \node [] (10) at (6, -0.25) {};
        \node [] (11) at (8, 0.25) {};
        \node [] (12) at (8, -0.25) {};
        \node [] (13) at (16.5, 0) {};
        \node [] (14) at (16, 0.25) {};
        \node [] (15) at (16, -0.25) {};
        \node [] (16) at (12, 0.25) {};
        \node [] (17) at (12, -0.25) {};
        
        \node [] (6m) at (0, -0.55) {$0$};
        \node [] (7m) at (4, 0.55) {$\{\sigma_1 \}$};
        \node [] (8m) at (4, -0.55) {$1/q_{\sigma_1}$};
        \node [] (9m) at (6, 0.55) {$\{\sigma_2 \}$};
        \node [] (10m) at (6, -0.55) {$1/q_{\sigma_2}$};
        \node [] (11m) at (8, 0.55) {$\{\sigma_1\}$};
        \node [] (12m) at (8, -0.55) {$2/q_{\sigma_1}$};
        \node [] (16m) at (12, 0.55) {$\{\sigma_1,\sigma_2\}$};
        \node [] (17m) at (12, -0.55) {$3/q_{\sigma_1}=2/q_{\sigma_{2}}$};
        \node [] (18m) at (14.5, -0.55) {$\dots$};
        \node [] (19m) at (14.5, 0.55) {$\dots$};

        \draw (0.center) to (3.center);
        \draw (5.center) to (6.center);
        \draw (7.center) to (8.center);
        \draw (9.center) to (10.center);
        \draw (11.center) to (12.center);
        \draw [bend left=75, looseness=0.00] (3.center) to (13.center);
        \draw (14.center) to (13.center);
        \draw (13.center) to (15.center);
        \draw (16.center) to (17.center); 
    \end{tikzpicture}
    \vspace{-2ex}
    \end{center}
    \caption{Constructing the string $s$ by distributing the different symbols evenly within the interval $(0,n]$ (here with $3/q_{\sigma_1}=2/q_{\sigma_{2}}$), and then collecting them from left to right.  }\label{fig:WalkTheLine}
    \end{figure}
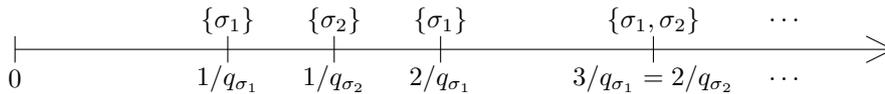
    
    It is then not too hard to convince yourself that this $s$ indeed satisfies the claim. Namely, for any $s' = (s'_1,\ldots,s'_m)$ as considered, we can again walk along the interval from $0$ and $n$, and we will then encounter all the symbols of $s'$, one by one: we will encounter the symbol $s'_1$ within the walk from $0$ to $1/q_{s'_1}$, the symbol $s'_2$ then within the walk from $1/q_{s'_1}$ to $1/q_{s'_1} + 1/q_{s'_2}$, etc. 

    Putting this idea into a formal proof is somewhat tedious, but in the end not too difficult. In order to formalize things properly, we generalize the standard notion of a sequence $s \in \Sigma^*$ in a way that allows us to talk about ``attaching'' a symbol to a point on $\mathbb{R}$, etc., in a rigorous way. 
    Formally, we define a {\em line sequence} to be an arbitrary finite (possibly empty) subset $S \subseteq \mathbb{R} \times \Sigma$, i.e., 
    $$
    S = \{(t_1,s_1),\ldots,(t_\ell,s_\ell)\} \in \TS{\mathbb{R}}{\Sigma} \, ,
    $$
    where w.l.o.g.\ we will always assume that $t_1 \leq \ldots \leq t_\ell$. 
    We may think of the symbol $s_i$ to ``occur at the time'' $t_i$.\footnote{Note that we allow $t_i = t_j$ for $i \neq j$ while the definition prohibits $(t_i,s_i) = (t_j,s_j)$. If desired, one could allow the latter by letting $S$ be a multi-set, but this is not necessary for us. }
    For a subset $T \subset \mathbb{R}$, the set $\TS{T}{\Sigma}$ then obviously denotes the set of line sequences with $t_1,\ldots,t_\ell \in T$. 

    Assuming that the alphabet $\Sigma$ is equipped with a total order $\leq$, any line sequence $S = \{(t_1,s_1),\ldots,(t_\ell,s_\ell)\}$ is naturally associated with the ordinary sequence 
    $$
    \pi(S) := (s_1,\ldots,s_\ell) \in \Sigma^* \, ,
    $$
    which is uniquely determined by the convention $t_1 \leq \ldots \leq t_\ell$ and insisting on $s_i \leq s_j$ whenever $t_i = t_j$ for $i < j$. 

    This {\em projection} $\pi: \TS{\mathbb{R}}{\Sigma} \rightarrow \Sigma^*$ preserves the characteristic of the sequence, i.e., if $s = (s_1,\ldots,s_\ell) = \pi(S)$ then
    \begin{equation}\label{eq:preserve}
    \#\{t \,|\, (t,\sigma)\in S\} = \#\{i \,|\, s_i=\sigma\} 
    \end{equation}
    for any $\sigma \in \Sigma$. 
    Furthermore, for $T,T' \subset \mathbb{R}$ with $T < T'$ point-wise, and for $S \in \TS{T}{\Sigma}$ and $S' \in \TS{T'}{\Sigma}$, it is easy to see that 
    $
    \pi(S \cup S') =  \pi(S)\|\pi(S') \, ,
    $
    from which it then follows that for ordinary sequences $s,s' \in \Sigma^*$
    \begin{equation}\label{eq:concat}
    s \sqsubseteq \pi(S) \,\wedge\, s' \sqsubseteq \pi(S') \;\Longrightarrow\; s\|s' \sqsubseteq \pi(S)\|\pi(S')  = \pi(S \cup S') \, . 
    \end{equation}
    A final, simple observation, which follows directly from the definitions, is that for $\sigma \in \Sigma$, i.e.\ a sequence of length $m=1$, $\sigma \sqsubseteq \pi(S)$ holds if and only if there exists a time $t \in \mathbb{R}$ such that $(t,\sigma) \in S$.

    \begin{proof}[Proof of Lemma~\ref{lem:string}]
        For any symbol $\sigma \in \Sigma$ let $S_\sigma$ be a line sequence 
        $$
        S_\sigma := \big\{\textstyle\frac{1}{q_\sigma},...,\frac{nq_\sigma}{q_\sigma}\big\} \times \{\sigma\} \in \TS{(0,n]}{\Sigma} \, ,
        $$
        and set $S := \bigcup_{\sigma \in \Sigma} S_\sigma$. 
        We will show that $s := \pi(S)$ is as claimed. 
        
        The claim on the characteristic of $s$ follows from the preservation of the characteristic under $\pi$, i.e.\ (\ref{eq:preserve}), and from $\#\{t\,|\,(t,\sigma)\in S\} = \#S_\sigma = n\cdot q_\sigma$, which holds by construction of $S$. 
        
        Let $s' = (s'_1,\dots,s'_m)\in\Sigma^*$ be arbitrary with characteristic bounded by~$q$. 
        We consider the times $\tau_j := 1/q_{s'_1} + \cdots + 1/q_{s'_j}$ for $j \in \{1,\ldots,m\}$, and we let $T_j$ be the interval 
        $$
        T_j := \big(\tau_{j-1},\tau_j \big] = \big(\tau_{j-1},\tau_{j-1} \!+\! \textstyle\frac{1}{q'_j}\big] \subset \mathbb{R} \, ,
        $$
        and decompose $S = S_1 \cup \ldots \cup S_m$ with $S_j := S \cap (T_j \!\times\! \Sigma) \in \TS{T_j}{\Sigma}$. 
        Here, we exploit that
        $$
        \tau_m  = \sum_{\sigma \in \Sigma} \frac{\#\{i\,|\, s'_i = \sigma\}}{q_\sigma} \leq \sum_{\sigma \in \Sigma} \frac{q_\sigma}{q_\sigma} = n \, ,
        $$
        and so the $S_j$'s indeed cover all of $S \in \TS{(0,n]}{\Sigma}$. 
        Given that the interval $T_j \subset (0,n]$ has size $1/q_{s'_j}$, there exists a time $t_j \in T_j  \cap \big\{\textstyle\frac{1}{q_\sigma},...,\frac{nq_\sigma}{q_\sigma}\big\}$. But then, $(t_j,s'_j) \in S_j$ by construction of $S$, and therefore $s'_j \sqsubseteq \pi(S_j)$. Finally, since $T_{j-1} < T_j$, property (\ref{eq:concat}) implies that 
        $$
        s' = s'_1\| \cdots \|s'_m \sqsubseteq \pi(S_1 \cup \ldots \cup S_m) = s
        $$
    which was to be shown. 
    \end{proof}
While Lemma~\ref{lem:string} above settles the existence question, the following two observations settle the corresponding efficiency questions. For concreteness, we assume $\Sigma = \{1,\ldots,n\}$ below, and thus can identify the function $q:\Sigma\to\mathbb{N}$, $\sigma \mapsto q_\sigma$ with the vector ${\bf q} = (q_1,\ldots,q_n)$.  

First, we observe that the line sequence $S$ defined in the proof above, as well as its projection $s = \pi(S)$, can be computed in polynomial time in $q_1+\cdots + q_n$; thus, we have the following. 
    \begin{lemma}\label{lem:efficienct_s}
        There exists a polynomial-time algorithm that, on input $1^{\bf q}$, computes a string $s\in\Sigma^*$ as specified in the proof of Lemma~\ref{lem:string}.
    \end{lemma}

Furthermore, for any $s'\in\Sigma^*$ with characteristic at most $q$, for which we then know by Lemma~\ref{lem:string} that $s'$ can be embedded into $s$, the following ensures that this embedding can be computed efficiently and {\em on the fly}. 

    \begin{lemma}\label{lem:online_emb}
        There exists a polynomial-time algorithm $\Ecal$ such that for every string $s\in\Sigma^*$ and every subsequence $s' =(s'_1,\dots,s'_m)\sqsubseteq s$, the following holds. Computing inductively $j_i\gets \Ecal(s,s'_i,j_{i-1})$ for every $i\in[m]$, where $j_0 : = 0$, results in an increasing sequence $j_1<\dots<j_m$ with
        $$
        s' = (s_{j_1},\dots,s_{j_m})\;.
        $$
    \end{lemma}
    
    The algorithm $\Ecal$ simply follows the obvious greedy strategy: for each $s'_i$ it looks for the next $j_i$ for which $s'_i = s_{j_i}$. More formally: 
    
    \begin{proof}
        The algorithm $\Ecal(s,s'_i,j_{i-1})$ computes 
        \begin{equation}\label{eq:embed_idx}
            j_{i}:=\min\left\{k\in{\mathbb{N}} \,|\, j_{i-1}<k\leq m, s_k = s'_{i} \right\}  \, .
        \end{equation}
It can be easily shown that the minimum is well-defined, i.e. taken over a non-empty set for each $i$ by the assumption that $s'$ is a subsequence of $s$, and thus by construction, every $j_i$ is such that $s'_i=s_{j_i}$ while keeping $j_1<\dots<j_n$ increasing. This concludes the proof.
    \end{proof}

    \subsection{Wrapping up the Proof of Theorem~\ref{thm:reduction}}\label{sec:Proof}
    
The claimed interactive oracle algorithm $\Bcal$ now works in the obvious way. 
On input ${\bf q}$ (provided in unary) and for any $\Acal$, $\Bcal[\Acal]$ will make static oracle queries to $\Ocal_{s_1},\Ocal_{s_2},\ldots,\Ocal_{s_{nQ}}$, where $s = (s_1,\dots,s_{nQ}) \in \{1,\ldots,n\}^*$ is the string promised to exist by Lemma~\ref{lem:string}, with $Q = q_1 + \cdots + q_n$. 
In more detail, it first computes $s$ using the algorithm from Lemma~\ref{lem:efficienct_s}. 
Then, for the $i$-th oracle query that $\Bcal$ receives from $\Acal$ (starting with $i = 1$), and which consists of the identifier $s'_i  \in \{1,\ldots,n\}$ of which oracle to query now and of the actual input to the oracle $\Ocal_{s'_i}$, the algorithm $\Bcal$ does the following: it computes $j_i\gets \Ecal(s,s'_i,j_{i-1})$ using the algorithm from Lemma~\ref{lem:online_emb}, makes dummy queries to $\Ocal_{s_{j_{i-1}+1}},\ldots,\Ocal_{s_{j_i-1}}$, and forwards $\Acal$'s query input to $\Ocal_{s_{j_i}} = \Ocal_{s'_i}$. 
The fact that $(j_1,\ldots,j_Q)$ computed this way forms an embedding of $s' = (s'_1,\ldots,s'_Q)$ into $s$ ensures that $\Bcal$ is able to forward all the queries that $\Acal$ makes to the right oracle, and so $\Acal$ will produce its output as in an ordinary run with direct adaptive access to the oracles.

%
%
%
%

    \subsection{Applications}
    
    To demonstrate the usefulness of our adaptive-to-static compiler, we briefly discuss three results from the literature. For two of them, the adaptivity of the attacker was explicitly declared as an obstacle in the security proof, and dealing with it complicated the proof substantially. These complications could be avoided/removed by means of our adaptive-to-static compiler. For the third one, we can immediately strengthen one of the results, which is restricted to hold for static multi-oracle adversaries, by dropping this restriction via our compiler.



 	\subsubsection{Quantum Security of qTESLA.}
 	Our first application is in the context of qTESLA \cite{qTesla}, which is a signature scheme that made it into the second round of the NIST post-quantum competition. Its security is based on the Ring-LWE problem, to which the authors of \cite{RevqTesla} give a reduction in the quantum random-oracle model (QROM).%
	\footnote{We note that some versions of qTESLA have been broken \cite{LS19}, but the attack only applies to an optimized variant that was developed for the NIST-competition, and does not apply to the scheme in \cite{RevqTesla} that we discuss here.} 
	In the reduction, which starts from the security notion of {\em Unforgeability under Chosen Message Attack} (UF-CMA), the adversary can query a random-oracle $H$ as well as a signing oracle, where the order of oracle queries may be adaptive.
 	
 	The reduction strategy of \cite{RevqTesla} applies only to a static adversary, with a fixed query pattern. Thus, the authors first compile the adaptive into a naive static attacker by letting it do $q_H$ (the number of $H$-queries of the original adaptive adversary) $H$-queries between any two signing queries. Leaving it with this would blow up the number of $H$-queries to $q_S q_H$. In order to avoid that, they give the attacker a ``{\em live-switch}'', meaning that each query to $H$ may be in superposition of making the query and not making the query, and the total ``{\em query magnitude}'' on actual $H$-queries is still restricted to $q_H$. Not so surprising, adding even more ``quantumness'' to the problem in this way, makes the analysis more complicated (compared to using standard ``all-or-nothing'' static queries and a standard classical bound on the query complexity), but it allows the authors to avoid the above blow-up in the (classical) query complexity to transpire into the security loss. 
	 The overall loss they obtain in the end is 
	 $
 	O((q_Sq_H^2 + q_S^3 + q_S^2q_H)\cdot \epsilon)
 	$
 	for small $\epsilon$ determined by the parameters of the scheme.
 	
 	Since the security reduction in \cite{RevqTesla} intertwines the adaptive to static hurdle with other aspects of the proof, we cannot simply insert our Theorem \ref{thm:reduction} and then continue the proof as is. Still, by applying our result, we could obtain a static adversary with almost no cost in the number of $H$-queries, avoiding the need for the rather complicated ``live-switch superposition'' attacker, thus simplifying the overall proof significantly. Furthermore, looking ahead at Section~\ref{sec:skPRF}, our result allows us to obtain the much better $O(\sqrt{q_Oq_H^2\epsilon} + \sqrt{q_O^2q_H\epsilon})$ loss in a similar context\,---\,similar in the sense that it also involves two oracles where one reprograms the other at some high-entropy input. The adaptive to static reduction there allows us to apply some additional QROM tools that could potentially also be applied in the setting of qTESLA to improve the bound. However, actually doing this would require us to rewrite the entire proof of \cite{RevqTesla}, which we consider outside the scope of this work.

    \subsubsection{Quantum Security of the Function FX}
     Our second application is to \cite{JST21}, where the post-quantum security of the FX key-length extension is studied (which is a generalization of the Even-Mansour cipher). In a first part, post-quantum security of FX is shown under the restriction that the  inputs to the queries are fixed in advance. In a second part, towards avoiding this restriction, the authors consider a variation of the FX construction, which they call FFX (for ``function~FX''), and they show in their Theorem~3 post-quantum security of FFX under the restriction that the attacker is ``{\em order consistent}'', as they call it in \cite{JST21}, which is precisely our notion of a {\em static} multi-oracle algorithm. 
     Thus, by a direct application of our Theorem~\ref{thm:reduction}, this restriction can be dropped (almost) for free, i.e., with a small constant blow-up on the attackers advantage.

%

    \subsubsection{Quantum Security of the Even-Mansour Cipher.}
    The recent work \cite{ABKM21} shows full post-quantum security of the (unmodified) Even-Mansour cipher. 
    Is in the case of qTESLA, 
    the fact that the attacker can choose adaptively whether to query the public permutation of the cipher complicates the proof. Indeed, as is explained on page~3 in \cite{ABKM21}, this adaptivity issue forces the authors to extend the blinding lemma of Alagic {\em et al.} to a variant that gives a bound in terms of the \emph{expected} number of queries. While the authors succeed in providing such an extended version of the blinding lemma (Lemma~3 in \cite{ABKM21}), it further increases the complexity of an already involved proof.%
    \footnote{To be fully precise, Lemma~3 in \cite{ABKM21} also generalizes the original blinding lemma in a different direction by allowing to reprogram to an arbitrary value instead of a uniformly random one; however, this generalization comes for free in that the original proof still applies up to obvious changes, while allowing an expected number of queries, which is needed to deal with the adaptivity issue, requires a new proof. }
    
    Thus, again, our Theorem \ref{thm:reduction} could be used to simplify the given proof by bypassing the complications that arise due to the attacker choosing adaptively which oracle to query at what point.

%

    \section{Quantum Security of a Split-key PRF}\label{sec:skPRF}
    
    \subsection{Hybrid Security and skPRFs}

    A {\em split-key pseudorandom function} (skPRF), as introduced in \cite{giacon2018kem}, is a polynomial-time computable function  $F:\Kcal_1\times\dots\times\Kcal_n\times\Xcal\to\Ycal$ that is a pseudorandom function (PRF) in the standard sense {\em for every $i\in[n]$} when considered as a keyed function with key space $\Kcal_i$ and message space $\Kcal_1\times\dots\times\Kcal_{i-1}\times \Kcal_{i+1}\times\dots\times\Kcal_{n}\times\Xcal$, with the additional restriction that the distinguisher $\Acal$ (in the standard PRF security definition) must use a fresh $x \in \Xcal$ in every query $(k_1,\ldots,k_{i-1},k_{i+1},\ldots,k_n,x)$.

This restriction on the PRF distinguisher may look artificial, but is motivated by this definition of a skPRF being good enough for the intended purpose of a skPRF, namely to give rise to a {\em secure KEM combiner}. 
Indeed, \cite{giacon2018kem} shows that the naturally combined KEM, obtained by  
    concatenating the individual ciphertexts to $C=(c_1,\dots,c_n)$, and combining the individual session keys $k_1,\dots,k_n$ using the above mentioned skPRF as 
    $$
    K = F(k_1,\dots,k_n,C) \, ,
    $$
    is IND-CCA secure if at least one of the individual KEM's is IND-CCA secure. 
    
    The paper~\cite{giacon2018kem} also proposes a particularly efficient hash-based construction, given by 
    \begin{equation}\label{eq:F}
    F(k_1,\dots,k_n,x):=H(g(k_1,\dots,k_n),x)
    \end{equation}
    where $g:\Kcal_1\times\dots\times\Kcal_n\to\Wcal$ is a polynomial-time mapping with the property that, for some small $\epsilon$, 
        \begin{equation}\label{eq:skPRF_high_min_entropy}
        \Pr_{k_i\gets\Kcal_i}\left[g(k_1,\dots,k_n)=w\right]\leq\epsilon\;,
    \end{equation}
    for every $i \in [n]$ and for every $k_1,\dots,k_{i-1},k_{i+1},\ldots,k_n$ and every $w$; furthermore, $H:\Wcal\to\Ycal$ is a cryptographic hash function. Simple choices for the function $g$ are $g(k_1,\dots,k_n) = (k_1,\dots,k_n)$ and $g(k_1,\dots,k_n) = k_1 + \cdots  + k_n$. 
    
     It is shown in \cite{giacon2018kem} that this construction is a skPRF when $H$ is modelled as a random-oracle; indeed, it is shown that the distinguishing advantage is upper-bounded by $q_H\epsilon$, where $q_H$ is the number of queries to the random-oracle~$H$. 
     
     Given the natural use of combiners in the context of the upcoming transition to post-quantum cryptography, it is natural\,---\,and well-motivated\,---\,to ask whether $F$ can be proven to be a skPRF in the presence of a {\em quantum attacker}, i.e., when $H$ is modeled as a {\em quantum} random-oracle. Below, we answer this in the affirmative.

    \subsection{Quantum-security of the skPRF}\phantom{\cite{}}
    The goal of this section is to show the security of the skPRF (\ref{eq:F}) in the quantum random-oracle model. 
    In essence, this requires proving that $F$ is a PRF (in the quantum random-oracle model) with respect to {\em any} of the $k_i$'s being the key, subject to the restriction of asking a fresh $x$ in each query. 
    
    To simplify the notation, we fix the index $i \in [n]$ and simply write $k$ for $k_i$ and $x$ for $(k_1,\dots,k_{i-1},k_{i+1},\ldots,k_n,x)$, and we abstract away the properties of the function $g$ as follows. We let $$F(k,x) := H(h(k,x)) \, ,$$ where $h:\Kcal\times\Xcal\to\Wcal$ is an arbitrary function with the property that, for some parameter $\epsilon>0$,  
    \begin{equation}
        \Pr_{k\gets\Kcal}\left[h(k,x)=w\right]\leq\epsilon\; \label{eq:high_min_entropy} 
    \end{equation}
    for all $w\in\Wcal$ and $x\in\Xcal$. 
    Furthermore, in the PRF security game, we restrict the attacker/distinguisher $\Acal$ to queries $x$ with a fresh value of $h(k,x)$, no matter what $k$ is. 
    
    More formally, let $\Acal^{\Hcal,\Ocal}$ be an arbitrary quantum oracle algorithm, making quantum superposition queries to an oracle $\Hcal$ and classical queries to another oracle $\Ocal$, with the restriction that for every query $x$ to $\Ocal$ it holds that 
    \begin{equation}
        h(\kappa,x) \neq h(\kappa,x') \;,\label{eq:restrict_query}
    \end{equation}
    for any prior query $x'$ to $\Ocal$ and all $\kappa\in\Kcal$. 
     For any such oracle algorithm $\Acal^{\Hcal,\Ocal}$, we consider the standard PRF security games 
     $$
     {\sf PR}^1 := \Acal^{H,F} \qquad\text{and}\qquad {\sf PR}^0:=\Acal^{H,R} \, ,
     $$ 
     obtained by instantiating $\Hcal$ with a random function $H$ (the random-oracle) in both games, and in one game we instantiate $\Ocal$ with the pseudorandom function $F$, which we understand to return $F(k,x)$ on query $x$ for a random $k \gets\Kcal$, chosen once and for all queries, and in the other we instantiate $\Ocal$ with a truly random function $R$ instead. 
     
    We show that the distinguishing advantage for these two games is bounded as follows.

    \begin{theorem}\label{thm:pseudorandom_adaptive}
    Let $\Acal^{\Hcal,\Ocal}$ be a $(q_\Hcal,q_\Ocal)$-query oracle algorithm satisfying (\ref{eq:restrict_query}). Then
    $$
    \left|\Pr\left[1\gets {\sf PR}^1\right]
    - \Pr\left[1\gets {\sf PR}^{0}\right]\right|
    \leq 4 \sqrt{2q_\Ocal^2q_\Hcal\epsilon} + 4 \sqrt{2q_\Hcal^2q_\Ocal\epsilon}\;.
    $$
    \end{theorem}
    
    We can now apply Theorem~\ref{thm:pseudorandom_adaptive} to the function $h(k,x):=(g(k_1,\dots,k_n),\tilde x)$, where $k:=k_i$ and $x:=(k_1,\dots,k_{i-1},k_{i+1},\dots,k_n,\tilde x)$. Indeed, the condition  (\ref{eq:skPRF_high_min_entropy}) on $g$ implies the corresponding condition (\ref{eq:restrict_query}) on $h$, and the restriction on $\tilde x$ being fresh in the original skPRF definition implies the above restriction on $h(k,x)$ being fresh no matter what $k$ is, i.e.,~\ref{eq:high_min_entropy}). 
    Thus, we obtain the following. 
    
\begin{corollary}
For any function $g$ satisfying (\ref{eq:skPRF_high_min_entropy}) for a given $\epsilon > 0$, the function $F(k_1,\dots,k_n, x) := H(g(k_1,\dots,k_n), x)$ is a skPRF in the quantum random-oracle model with distinguishing advantage at most $4 \sqrt{2q_\Ocal^2q_\Hcal\epsilon} + 4 \sqrt{2q_\Hcal^2q_\Ocal\epsilon}$. 
\end{corollary}

    \subsection{Proof of Theorem~\ref{thm:pseudorandom_adaptive}}

    \begin{proof}[Proof (of Theorem~\ref{thm:pseudorandom_adaptive})]
        Let $\Acal^{{\Hcal},{\Ocal}}$ be an oracle algorithm as considered in the previous subsection. Thanks to Theorem~\ref{thm:reduction}, taking a factor-$2$ blow-up in the query complexity into account, we may assume $\Acal$ to be a {\em static} $(q_\Hcal,q_\Ocal)$-query oracle algorithm. 
        It will be convenient to write such a static algorithm as 
        $$
        \Acal^{[\mathbfcal{H}_0 \Ocal \mathbfcal{H}_1 \Ocal \mathbfcal{H}_2 \dots \Ocal \mathbfcal{H}_{q_\Ocal}]} \, ,
        $$ 
        where each block $\mathbfcal{H}_i = \Hcal \cdots \Hcal$ consists of a (possibly empty) sequence of symbols $\Hcal$ of length $q^\Hcal_{i} = |\mathbfcal{H}_i|$, and with the understanding that $\Acal$ first makes $q^\Hcal_{0}$ queries to $\Hcal$, then a query to $\Ocal$, then $q^\Hcal_{1}$ queries to $\Hcal$, etc., where, obviously, $q^\Hcal_{0} + \cdots + q^\Hcal_{q_\Ocal} = q_\Hcal$ then. Instantiating $\Hcal$ with $H$, and $\Ocal$ with $F$ and $R$, respectively, we can then write
        $$
        {\sf PR}^0 = \Acal^{[{\bf H}_0R{\bf H}_1\dots R{\bf H}_{q_\Ocal}]}
        \qquad\text{and}\qquad
        {\sf PR}^1 = \Acal^{[{\bf H}_0F{\bf H}_1\dots F{\bf H}_{q_\Ocal}]} \, .
        $$
        For the proof, we introduce certain hybrid games. For this purpose, we introduce the following alternative ({\em stateful} and $R$-dependent) instantiation $H'$ of~$\Hcal$. To start with, $H'$ is set to be equal to $H$, but whenever $R$ is queried on some input~$x$, $H'$ is {\em reprogrammed} at the point $h(k,x)$ to the value $H'(h(k,x)):=R(x)$. 
    For any $i$, we now define the two hybrid games
    \begin{align*}
        {\sf PR}^2_{i}
        :=&\Acal^{[{\bf H}_0 R\dots R{\bf H}_{i}{\color{red} F}{\bf H}_{i+1}'F\dots F{\bf H}_{q_\Ocal}']}\\
        \widetilde{\sf PR}^2_{i}
        :=&\Acal^{[{\bf H}_0 R\dots R{\bf H}_{i}{\color{red} R}{\color{blue}{\bf H}_{i+1}'}F\dots F{\bf H}_{q_\Ocal}']}
    \end{align*}
    and also spell out    
    \begin{align*}
        {\sf PR}^2_{i+1}
        =&\Acal^{[{\bf H}_0 R\dots R{\bf H}_{i} R{\color{blue}{\bf H}_{i+1}} F\dots F{\bf H}_{q_\Ocal}']} 
    \end{align*}
to emphasize its relation to $\widetilde{\sf PR}^2_{i}$. We note that in all of the above, the first occurrences of $\Hcal$ and $\Ocal$ are instantiated with $R$ and $H$, respectively, but at some point we switch to $R$ and $H'$ instead. 
 
 The extreme cases match up the games we are interested in. Indeed, 
    $$
    {\sf PR}^2_0 = \Acal^{[{\bf H}_0 F {\bf H}_1'\dots F{\bf H}_{q_\Ocal}']} = \Acal^{[{\bf H}_0 F {\bf H}_1\dots F{\bf H}_{q_\Ocal}]} = {\sf PR}^1 \, ,
    $$
    where we exploit that there are no queries to $R$ and thus $H'$ remains equal to $H$, and, by definition, 
    $$
    {\sf PR}^2_{q_\Ocal} =\Acal^{[{\bf H}_0R{\bf H}_1\dots R{\bf H}_{q_\Ocal}]} = {\sf PR}^0\;.
    $$
    Our goal is to prove the closeness of the following games
    $$
    {\sf PR}^1
    = {\sf PR}^2_0
    \approx \widetilde{{\sf PR}^{2}_0}
    \approx {\sf PR}^{2}_1
    \dots
    \approx{\sf PR}^2_{q_\Ocal-1}
    \approx \widetilde{\sf PR}^{2}_{q_\Ocal-1}
    \approx {\sf PR}^{2}_{q_\Ocal}
    = {\sf PR}^0\;. 
    $$
    We do this by means of applying Lemma~\ref{lem:hybrid_F2R} and~\ref{lem:hybrid_unreprogram}, which we state here and prove further down.     
    
    \begin{lemma}\label{lem:hybrid_F2R}
        For each $0\leq i<q_\Ocal$,
        \begin{align*}
            \left|\Pr\Bigl[1\gets{\sf PR}^2_i\Bigr]
            - \Pr\Bigl[1\gets\widetilde{\sf PR}^2_i\Bigr]\right|
            \leq 2\sqrt{\sum_{1\leq j\leq i}q^\Hcal_j \epsilon}\;.
        \end{align*}
    \end{lemma}
    \begin{lemma}\label{lem:hybrid_unreprogram}
        For each $0\leq i<q_\Ocal$,
        \begin{align*}
            \left|\Pr\Bigl[1\gets\widetilde{\sf PR}^2_i\Bigr]
            - \Pr\Bigl[1\gets{\sf PR}^2_{i+1}\Bigr]\right|
            \leq 2 q^\Hcal_{i+1} \sqrt{q_\Ocal\epsilon} \;.
        \end{align*}
    \end{lemma}
    
    Indeed, by repeated applications of these lemmas, and additionally using that $q^\Hcal_0  + \cdots + q^\Hcal_i  \leq q_\Hcal$ for all $0 \leq i \leq q_\Ocal$, we obtain 
    \begin{align*}
    \left|\Pr\left[1\gets{\sf PR}^1\right] - \Pr\left[1\gets{\sf PR}^0\right]\right|
    &\leq  2 \sum_{i=0}^{q_\Ocal}  \sqrt{\sum_{1\leq j \leq i}q^\Hcal_j \epsilon} +  2 \sum_{i=0}^{q_\Ocal} q^\Hcal_{i+1} \sqrt{q_\Ocal\epsilon} \\[0.5ex]
    &\leq 2\sqrt{q_\Ocal^2q_\Hcal\epsilon} + 2\sqrt{q_\Hcal^2q_\Ocal\epsilon}
    \end{align*}
    which concludes the claim of Theorem~\ref{thm:pseudorandom_adaptive} when incorporating the factor-2 increase in $q_\Hcal$ and $q_\Ocal$ due to switching to a static $\Acal$. 
    \end{proof}

It remains to prove Lemma~\ref{lem:hybrid_F2R} and~\ref{lem:hybrid_unreprogram}, which we do below. 
In both proofs, we use the {\em gentle measurement lemma} \cite[Lemma~9.4.1]{wilde2011classical}, which states that if a projective measurement has a very likely outcome then the measurement causes only little disturbance on the state. More formally, for any density operator $\rho$ and any projector $P$, where $p := \tr(P \rho P)$ then is the probability to observe the outcome associated with $P$ when measured using the measurement $\{P, \mathbb{I}- P \}$, the trace distance between the original state $\rho$ and the post-measurement state $\rho' := P \rho P/p$ is bounded by $\sqrt{1-p}$. This in turn implies that $\rho$ and $\rho'$ can be distinguished with an advantage $\sqrt{1-p}$ only. 
    
The proof of Lemma~\ref{lem:hybrid_F2R} additionally makes use of Zhandry's compressed oracle technique \cite{Zha19}. It is out of scope of this work to give a self-contained description of this technique; we refer to the original work~\cite{Zha19} instead, or to \cite{CFHL21}, which offers an alternative concise description. At the core is the observation that one can {\em purify} the random choice of the function $H$ and then, by switching to the Fourier basis and doing a suitable measurement, one can check whether a certain input $x$ has been ``recorded'' in the database (mind though that such a measurement disturbs the state). If the outcome is negative then the oracle is still in a uniform superposition over all possible hash values for $x$, and as a consequence, when removing the purification by doing a full measurement of $H$ (in the computational basis), $H(x)$ is ensured to be a ``fresh'' uniformly random value, with no information on $H(x)$ having been leaked in prior queries. 

In the proof of Lemma~\ref{lem:hybrid_F2R}, we use this technique to check whether {\em prior} to the crucial query, which is to $F$ in one and to $R$ in the other game, there was a query to $H$ that would reveal the difference, and we use (\ref{eq:high_min_entropy}) to argue that it is unlikely that such a query occurred. Since this measurement has a likely outcome, it is also ensured by the gentle measurement lemma that this measurement causes little disturbance.   
    

    \begin{proof}[Proof (of Lemma~\ref{lem:hybrid_F2R})]
        For convenience, we refer to the {\em crucial query} as the respective query to $F$ and $R$ that differs between
$$
        {\sf PR}^2_{i}
        =\Acal^{[{\bf H}_0 R\dots R{\bf H}_{i}{\color{red} F}{\bf H}_{i+1}'F\dots F{\bf H}_{q_\Ocal}']} \quad\text{and}\quad \widetilde{\sf PR}^2_i = \Acal^{[{\bf H}_0 R\dots R{\bf H}_{i}{\color{red} R}{\bf H}_{i+1}'F\dots F{\bf H}_{q_\Ocal}']} \, .
        $$
        Furthermore, we let $x$ be the input to that query, and we set $w := h(k,x)$, with $k$ being the key chosen and used by $F$. Note that up to this very query, the two games are identical. Also, by (\ref{eq:restrict_query}) it is ensured that for any prior query $x'$ to $R$ it holds that $h(k,x') \neq w$. 

        First, we consider the games ${\bf G}^{1}$ and $\widetilde{\bf G}^1$ that work exactly as ${\sf PR}^2_i$ and $\widetilde{\sf PR}^2_i$, respectively, except that, at the beginning of the games we set up the compressed oracle and answer all queries made to $H$ prior to the crucial query using the compressed oracle. Then, once $x$ is received during the crucial query, we do a full measurement of the purified (i.e.~uncompressed) oracle in order to obtain the function $H$, which is then to be used in the remainder of the games.  
        We note that setting up the function $H'$ is then necessarily also deferred to after this measurement, where $H'$ is then set to be equal to $H$, except that for any prior query $x'$ to $R$ it is reprogrammed to $H'(h(k,x')) := R(x')$. 
        Only once $H$ has been measured and $H'$ set up as above, is the crucial query then actually answered. 
        
        It follows from basic properties of the compressed oracle that the respective output distributions of ${\bf G}^{1}$ and $\widetilde{\bf G}^1$ match with those of ${\sf PR}^2_i$ and \smash{$\widetilde{\sf PR}^2_i$}. 
        
        Then, we define ${\bf G}^{2}$ and $\widetilde{\bf G}^2$ from ${\bf G}^{1}$ and $\widetilde{\bf G}^1$, respectively, by introducing one more measurement. Namely, right after $x$ is sent by $\Acal$ and before $H$ is measured, we measure in the compressed oracle whether the input $w = h(k,x)$ has been recorded in the database, 
        and in case of a positive outcome, the game aborts. By the gentle measurement lemma (and basic properties of the trace distance), 
        \begin{align*}
            &\left|\Pr\left[1\gets{\bf G}^{1}\right]
            - \Pr\left[1\gets {\bf G}^{2}\right]\right|
            \leq \sqrt{\Pr\left[{\bf G}^{2}\text{ aborts}\right]} 
        \end{align*}
        and similarly for $\widetilde{\bf G}^{1}$ and $\widetilde{\bf G}^{2}$, where $\widetilde{\bf G}^{2}$ aborts with the same probability as ${\bf G}^{2}$. 
        
        By basic properties, after $t:= q^\Hcal_0 + \cdots + q^\Hcal_i$ queries to the compressed oracle, no more than $t$ values have been recorded. I.e., if we were to measure, for the sake of the argument, the entire compressed oracle to obtain the full database $D$, it would hold that $\mathrm{supp}(D) := \{ u \,|\, D(u) \!\neq\! \bot \}$ has cardinality at most $t$. 
        Since $k$ has not been used yet and so is still freshly random (i.e., independent of $x$ and $D$), the high-entropy condition (\ref{eq:high_min_entropy}) then ensures that

        $$
        \Pr\bigl[\widetilde{\bf G}^2\text{ abort}\bigr]=\Pr\bigl[{\bf G}^{2}\text{ abort}\bigr] = \Pr\bigl[ w \in \mathrm{supp}(D) \bigr] \leq \sum_{j<i}q^\Hcal_j\epsilon \, .
        $$
   
        
        It remains to show that ${\bf G}^2$ and $\widetilde{\bf G}^2$ behave identically conditioned on not aborting. 
        The only difference between the two games is that in ${\bf G}^2$ the crucial query is answered with $y := H(h(k,x)) = H(w)$ and $H'$ is {\em not} reprogrammed at the point $w$, while in $\widetilde{\bf G}^2$ the crucial query is answered with $y := R(x)$ and $H'$ {\em is} reprogrammed at the point $w$ to $H'(w) := R(x)$. 
        We argue that this difference is not noticable by $\Acal$. 
        
        First, we note that $y$ is a fresh random value in both games. In the former game it is because, conditioned on not aborting, the compressed oracle at the register $h(k,x)$ is $\bot$, and so when uncompressing and measuring to obtain $H$, the hash value $H(w)$ will be a fresh random value.  In the latter game it is because $R(x)$ is a truly random function and, due to (\ref{eq:restrict_query}), $x$ has not been queried to $R$ before. 
        
        Second, we observe that $y = H'(w)$ in both games. Indeed, in $\widetilde{\bf G}^2$ this holds by definition; in ${\bf G}^2$ it holds because $H'(w) = H(w)$, which follows from the fact that $H'$ is reprogrammed only at points $w' = h(k,x')$ with $x'$ being a prior query to $R$, but then (\ref{eq:restrict_query}) ensures that $w' \neq w$. 
        
        Thus, in both games, from $\Acal$'s perspective, the tuple $(k,y,H',H\backslash w)$ of random variables has the same distribution, where $H\backslash w$ refers to the function (table of) $H$ but with the value at the point $w$ removed. The only difference is that in one game $H'(w) = H(w)$ and in the other not (necessarily). However, the future behavior of $\Acal$ in both games only depends on $(k,y,H',H\backslash w)$, and thus $\Acal$ behaves the same way in both games. 
        Here we are exploiting that the future hash queries by $\Acal$ are to $H'$ (and not to $H$ anymore), and, once more, we are using the restriction (\ref{eq:restrict_query}), here to ensure that for any future $F$-query $x'$ by $\Acal$, it holds that $h(k,x') \neq w$, and thus the response does not depend on $H(w)$. Thus, $H(w)$ does indeed not affect $\Acal$'s behavior after the crucial query.

Exploiting that  ${\sf PR}^2_i = {\bf G}^{1} \approx {\bf G}^{2} = \widetilde{\bf G}^2 \approx \widetilde{\bf G}^1 = \widetilde{\sf PR}^2_i$, with the approximations bounded as discussed further up, we obtain the claimed closeness claim.          
        This concludes the proof.
    \end{proof}
    
    \begin{proof}[Proof of Lemma~\ref{lem:hybrid_unreprogram}]
        In order to show the closeness between $\widetilde{\sf PR}^2_i$ and ${\sf PR}^2_{i+1}$, we define the intermediate games 
        \begin{align*}
            {\bf G}_{i,j}:=\Acal^{[{\bf H}_0 R\dots {\bf H}_i R {\color{blue}{\bf H}_{i,j}' {\bf H}_{i,j}}F\dots F{\bf H}'_{q_\Ocal}]} 
        \end{align*}
         for $0\leq j\leq m:=q^\Hcal_{i+1}$, where ${\color{blue}{\bf H}_{i,j}'}$ and ${\color{blue}{\bf H}_{i,j}}$ consists of $j$ and $m-j$ copies of $H'$ and $H$ respectively. Note that for the extrame cases we have 
        $$
        {\bf G}_{i,0}=\widetilde{\sf PR}^2_i \qquad\text{and}\qquad {\bf G}_{i,m}={\sf PR}^2_{i+1} \, .
        $$
        Thus, it suffices to show closeness between ${\bf G}_{i,j}$ and ${\bf G}_{i,j+1}$ for any $0\leq j<m$. Note that they only differ at one query, which is either to $H'$ or to $H$, which we will refer to as the {\em crucial query} for convenience. In the remainder, $i$ and $j$ are arbitrary (in the considered ranges) but fixed. 

        Define the games $\widetilde{\bf G}^{1}$ and ${\bf G}^1$ from ${\bf G}_{i,j}$ and ${\bf G}_{i,j+1}$ respectively as follows.
        Let $X$ be the set of queries $x$ made to $R$ prior to the crucial query, and set $S := \{ h(k,x) \,|\, x \!\in\! X \}$. We then measure the crucial query, which may be in a superposition, with the binary measurement that checks whether the crucial query is an element of $S$, and we abort if this is the case. 
        
        In case of a negative outcome, i.e., the crucial query is {\em not} in $S$, there is no difference between the reply provided by $H$ and by $H'$, and thus there is no difference between the two games\,---\,and in case of a positive outcome, they both abort. 
        In order to argue that this measurement causes little disturbance, we again use the gentle measurement lemma to argue that 
                \begin{align*}
            &\left|\Pr\left[1\gets {\bf G}^{1}\right]
            - \Pr\left[1\gets{\bf G}_{i,j+1}\right]\right|
            \leq \sqrt{\Pr\left[{\bf G}^{1}\text{ abort}\right]}\;,
        \end{align*}
        and correspondingly for ${\bf G}_{i,j}$ and $\widetilde{\bf G}^{1}$. 
        So it remains to bound the abort probability. 
        For the purpose of the argument, let us do a full measurement of the query, and let $w$ be the outcome. We note that $k$ has not been used yet, and thus remains a fresh random key, independent of $w$ and $X$. Thus, using (\ref{eq:high_min_entropy}),
        $$
                \Pr\left[{\bf G}^1\text{ abort}\right]=\Pr\left[\widetilde{\bf G}^{1}\text{ abort}\right]
        =\Pr\left[w \in S \right]
        \leq \sum_{x \in X} \Pr\left[ w = h(k,x) \right]
        \leq q_\Ocal\epsilon \;.
        $$
        Adding up this error term over the sequence ${\bf G}_{i,0}\approx\dots\approx{\bf G}_{i,m}$ of approximations, the proof is concluded.
%
    \end{proof}

    \section*{Acknowledgments}
    
    JD was funded by the ERC-ADG project ALGSTRONGCRYPTO (project number 740972). YHH was funded by the Dutch Research Agenda (NWA) project HAPKIDO (project number NWA.1215.18.002), which is financed by the Dutch Research Council (NWO).

    \ifnum\submission=0
        \bibliographystyle{alpha}
    \else
        \bibliographystyle{splncs04}
    \fi
    \bibliography{ref}

\end{document}